\documentclass[letterpaper, 10 pt, conference]{ieeeconf}  

\IEEEoverridecommandlockouts                              
\usepackage{graphics} 
\usepackage{epsfig} 
\usepackage{times} 
\usepackage{amsmath} 
\usepackage{amssymb}  
\usepackage{subfigure}
\usepackage{cite}
\usepackage{amsfonts}
\usepackage{fancybox}
\usepackage{relsize}
\usepackage{floatrow}
\usepackage{colortbl}
\usepackage{hhline}
\usepackage{lipsum}
\usepackage{mathtools}
\usepackage{cuted}
\usepackage{caption}
\usepackage{romannum}
\usepackage{dsfont}
\usepackage{algorithm}
\usepackage{algorithmic}
\usepackage{tabu}
\usepackage{flushend}
\usepackage{multirow}

\newtheorem{thm1}{\bf Theorem}
\newtheorem{prop1}{\bf Proposition}
\newtheorem{lem1}{\bf Lemma}
\newtheorem{assmpt1}{\bf Assumption}
\newtheorem{defn17}{\bf Definition}
\newtheorem{rem1}{\bf Remark}

\newtheorem{cor1}{\bf Corollary}

\newenvironment{definition}{\begin{defn17}}{\hfill$\Diamond$\end{defn17}}
\newenvironment{remark}{\begin{rem1}}{\hfill$\Diamond$\end{rem1}}

\newenvironment{theorem}{\begin{thm1}}{\hfill$\Diamond$\end{thm1}}
\newenvironment{corollary}{\begin{cor1}}{\hfill$\Diamond$\end{cor1}}

\title{\LARGE \bf Uncertainty Quantification of Autoencoder-based Koopman Operator}

\author{Jin Sung Kim$^1$, Ying Shuai Quan$^1$, and Chung Choo Chung$^2$$^\dag$%
\thanks{ This work was supported by the National Research Foundation of Korea (NRF) Grant funded by the Ministry of Science and ICT (MSIT, Data-Driven Optimized Autonomous Driving Technology Using Open Set Classification Method) under Grant 2021R1A2C2009908. }
\thanks{$^1$J. S. Kim, Y. S. Quan are with the Dept. of Electrical Engineering, Hanyang University, Seoul 04763, Korea.
(e-mail: {\tt \{jskim06, ysquan\}@hanyang.ac.kr}}
\thanks{$^2$C. C. Chung is with the Div. of Electrical and Biomedical Engineering, Hanyang University, Seoul 04763, Korea. (+82-2-2220-1724, e-mail: {\tt cchung@hanyang.ac.kr})}
\thanks{\dag: Corresponding author 
}
}

\begin{document}

\maketitle
\thispagestyle{empty}
\pagestyle{empty}

\begin{abstract}
This paper proposes a method for uncertainty quantification of an autoencoder-based Koopman operator. The main challenge of using the Koopman operator is to design the basis functions for lifting the state. To this end, this paper builds an autoencoder to automatically search the optimal lifting basis functions with a given loss function. We approximate the Koopman operator in a finite-dimensional space with the autoencoder, while the approximated Koopman has an approximation uncertainty. To resolve the problem, we compute a robust positively invariant set for the approximated Koopman operator to consider the approximation error. Then, the decoder of the autoencoder is analyzed by robustness certification against approximation error using the Lipschitz constant in the reconstruction phase. The forced \emph{Van der Pol} model is used to show the validity of the proposed method. From the numerical simulation results, we confirmed that the trajectory of the true state stays in the uncertainty set centered by the reconstructed state.
\end{abstract}

\section{INTRODUCTION}
%
%
The data-driven approach in model identification has recently attracted widespread interest in many applications, such as fluid dynamical system~\cite{brunton2016discovering}, or vehicle dynamics~\cite{kim2023koopman}, because of its high nonlinearity.
The Koopman operator is one of the data-driven approaches to capture a nonlinear dynamical system in the form of a linear relationship in the infinite-dimensional space~\cite{koopman1931hamiltonian}. Since the Koopman operator has its linear property for representing the given system even though the underlying system is nonlinear, much research in recent years has focused on the Koopman-based modeling and control method~\cite{mauroy2020koopman}.
%

The Koopman operator-based model identification for the (controlled) nonlinear system has been studied. In~\cite{narasingam2020koopman, korda2018linear}, the Koopman was used to model the nonlinear dynamics with extended dynamic model decomposition (EDMD), and the Koopman-based Model Predictive Control (MPC) was studied.
%
%
%
%
In~\cite{narasingam2020koopman, korda2018linear},
%
the authors use some radial basis functions as lifting functions for expanding the original state space to a high-dimensional lifted space in obtaining the Koopman operator~\cite{mauroy2020koopman}.
In~\cite{takeishi2017learning, han2020deep, iacob2021deep}, however, it is reported that the selection strategy of appropriate lifting functions might only be systematically given to compute the Koopman operator with domain knowledge.
From the viewpoint of resolving the problem, recent work has been studied for constructing the optimal lifting function in~\cite{korda2020optimal}.
%

In recent years, machine learning and deep learning techniques have been actively used to construct lifting functions in the Koopman theory. In~\cite{takeishi2017learning, yeung2019learning, iacob2021deep, han2020deep}, the authors introduced the deep neural network as lifting functions of the Koopman operator.
In addition, \cite{otto2019linearly} used the autoencoder (AE) structure to lift and reconstruct the state of the system.
%
%
The main advantage of the learning-based Koopman model is to automatically search the lifting functions by training process without hands-on design.
In order to take advantage of the learning-based approach, we need to certify the robustness of the neural network to consider uncertainties when the neural network is used as the lifting function and reconstruction.
To take the issue into account, in~\cite{pan2020physics}, the authors studied the uncertainties of the AE-based Koopman operator with mean-field variational inference. Another scheme to quantify the uncertainties can be constructed by viewing the neural network as a nonlinear map~\cite{fazlyab2019efficient}.
%
%
In~\cite{fazlyab2019efficient, quan2022linear}, the authors proposed the estimation method for a Lipschitz constant of a neural network by regarding the activation functions of a hidden layer as quadratic constraints.
%
%
Therefore, the Koopman-basd model could be used to predict a state trajectory of a nonlinear system within an uncertain set by computing the Lipschitz constant of the neural network.

In this context,  this paper proposes the method for uncertainty quantification of an AE-based Koopman operator.
As mentioned above, one primary difficulty of using the Koopman operator is to design the basis lifting functions. To resolve the problem,
we construct an AE to automatically obtain the optimal lifting function and reconstruction map with respect to a given cost function.
%
%
In order to consider a practical implementation of the Koopman operator, the approximated Koopman operator is needed in finite-dimensional space, which results in a residual term.
Moreover, the residual term can be propagated through the decoder in the reconstruction phase.
Thus, this paper considers both the approximation and the reconstruction uncertainty. For this end, first, we compute the robust positively invariant set for the approximated Koopman model to consider the approximation uncertainty.
Moreover, we analyze the reconstruction error in the decoder with robustness certification using the Lipschitz constant, and the reconstruction uncertainty set is computed.
Finally, numerical simulation is conducted with the forced \emph{Van der Pol} to validate the effectiveness of the proposed method. From the simulation results,
%
%
it is confirmed that, given a dataset, the true state always stays within the uncertainty set centered by the reconstructed state.

\section{Koopman Operator Theory}

\subsection{Basic Concept of Koopman Operator}
Let us start with the Koopman operator approach for an discrete-time autonomous nonlinear dynamics
\begin{equation}
\boldsymbol{\eta}_{k+1} = f_a ( \boldsymbol{\eta}_k),
\label{eq:autonomous nonlinear model}
\end{equation}
where $\boldsymbol{\eta}_k \in \mathcal{N}$ is the state of the system, $f_a$ is a nonlinear function that evolves the state of the system forward in time, and $k \in \mathds{Z}_+$ is the discrete-time step.
Let us define a real-valued scalar function $\psi: \mathcal{N} \rightarrow \mathds{R}$, which is so-called \emph{observable}~\cite{mauroy2020koopman}. Each real-valued function $\psi$ is an element of an infinite-dimensional function space $\mathcal{F}_a$ (i.e., $\psi \in \mathcal{F}_a$)~\cite{mauroy2020koopman}. In the function space $\mathcal{F}_a$, the Koopman theory provides an alternative representation of~\eqref{eq:autonomous nonlinear model} by introducing the Koopman operator $\mathcal{K}_a: \mathcal{F}_a \rightarrow  \mathcal{F}_a$ defined by
\begin{equation}
\mathcal{K}_a  \psi (\boldsymbol{\eta}_k ) := \psi (f_a (\boldsymbol{\eta}_k) )
\end{equation}
for every $\psi \in \mathcal{F}_a$, where $\mathcal{F}_a$ is invariant under the Koopman operator~\cite{mauroy2020koopman, korda2018linear}.
The Koopman operator for autonomous nonlinear systems can be generalized to controlled nonlinear systems with a slight change~\cite{williams2016extending, korda2018linear}. This paper adopts the data-driven method in~\cite{korda2018linear}, which is practical and rigorous approach.
Consider a discrete-time nonlinear controlled system given as
\begin{equation}
\textbf{x}_{k+1} = f (\textbf{x}_{k}, \textbf{u}_k),
\label{eq:controlled nonlinear model}
\end{equation}
where $\textbf{x}_k \in \mathcal{X} \subseteq \mathds{R}^n$ is the state of the controlled system, $\textbf{u}_k \in \mathcal{U} \subseteq {\mathds R}^m$ is the input. One can define the extended state space $\mathcal{X} \times \ell (\mathcal{U})$, where $\ell (\mathcal{U})$ is the space of all control sequences, $\boldsymbol{\mu} := (\textbf{u}_k)_{k=0}^{\infty}$ with $\textbf{u}_k \in \mathcal{U}$. Then, in line with~\cite{korda2018linear}, the extended state is defined as
\begin{equation}
\chi = \begin{bmatrix} \textbf{x} \\ \boldsymbol{\mu} \end{bmatrix}.
\end{equation}
With this extended state, \eqref{eq:controlled nonlinear model} can be in the form of an autonomous system defined  by
\begin{equation}
\chi_{k+1} = F(\chi_k) :=
\begin{bmatrix}
f(\textbf{x}_k, \boldsymbol{\mu}_k (0) ) \\ \mathcal{S} \boldsymbol{\mu}_k
\end{bmatrix},
\label{eq:extended state model}
\end{equation}
where $\mathcal{S}$  is the left shift operator, i.e., $\mathcal{S} \boldsymbol{\mu}_k = \boldsymbol{\mu}_{k+1}$, and $\boldsymbol{\mu}_k (0) \in \mathds{R}^m$ is the first element of the control sequence at the time step $k$ (i.e., $\boldsymbol{\mu}_k(0) = \textbf{u}_k$)~\cite{korda2018linear}.
Then, we can define the Koopman operator $\mathcal{K}: \mathcal{F} \rightarrow  \mathcal{F}$ for~\eqref{eq:extended state model} as
\begin{equation}
\mathcal{K} \phi (\chi_k ) = \phi (F(\chi_k) ),
\label{eq:koopman definition}
\end{equation}
where $\phi: \mathcal{X} \times \ell(\mathcal{U}) \rightarrow \mathds{R}$ is a real-valued function, which belongs to the extended observables space $\mathcal{F}$~\cite{mauroy2020koopman}.
%
%
%
%
Interestingly, it can be seen that the Koopman operator is linear in the space $\mathcal{F}$ although the dynamical system is nonlinear~\cite{mauroy2020koopman}.
%

\subsection{Extended Dynamic Mode Decomposition for Approximation of Koopman Operator}
From the definition~\eqref{eq:koopman definition}, the Koopman operator $\mathcal{K}$ operates on $\mathcal{F}$, which needs an infinite number of basis functions. Thus, using the Koopman operator directly in the real world is not practical unless we can obtain the approximated finite-dimensional Koopman operator.
To this end, a finite-dimensional subspace $\bar{\mathcal{F}} \subset \mathcal{F}$ can be considered, which is spanned by a set of basis functions. In $\bar{\mathcal{F}}$, we can obtain a finite-dimensional Koopman operator $K \in \mathds{R}^{N \times N}$.
In general, however, $\bar{\mathcal{F}}$ is not invariant with regard to $\mathcal{K}$, which results in a residual term due to the approximation of the Koopman operator~\cite{mauroy2020koopman}.
The EDMD approach minimizes the residual term in the $l_2$ sense and is widely used to approximate the Koopman operator~\cite{williams2015data}.
The first step of the EDMD is to collect the data as
\begin{eqnarray}
\begin{array}{rcl}
\textbf{X} &=&
\begin{bmatrix}
\textbf{x}_1 & \textbf{x}_2 & \dots & \textbf{x}_{M}
\end{bmatrix} \in {\mathds R}^{n \times M},
\vspace{2mm} \\
\textbf{U} &=&
\begin{bmatrix}
\textbf{u}_1 & \textbf{u}_2 & \dots & \textbf{u}_{M}
\end{bmatrix} \in {\mathds R}^{m \times M},
\vspace{2mm} \\
\textbf{Y} &=&
\begin{bmatrix}
\textbf{y}_1 & \textbf{y}_2 & \dots & \textbf{y}_M
\end{bmatrix} \in {\mathds R}^{n \times M}.
%
%
\end{array}
\label{eq:data collection}
\end{eqnarray}
where $M$ is the number of data sample, and $\textbf{y}_k=\textbf{x}_{k+1}=f (\textbf{x}_k, \textbf{u}_k)$. Let us assume that we have basis functions $\phi_i$.
Then, the optimization problem of finding $K$ is
\begin{equation}
\min_{K}
\sum_{k=1}^M
\|
\boldsymbol{\phi} ( \chi_{k+1} ) -
K \boldsymbol{\phi} (\chi_{k})
\|_2^2.
\label{eq:minimization for koopman}
\end{equation}
where $\boldsymbol{\phi}=[ \phi_1 , \phi_2 ,  \cdots  ]^T$.
However, note that the extended state $\chi$ is of in general infinite-dimension, which leads to the fact that \eqref{eq:minimization for koopman} might not be computable. Hence, this paper designs a computable observable function as
\begin{equation}
\boldsymbol{\phi} ( \chi_k ) =
\begin{bmatrix}
\boldsymbol{\pi} ( \textbf{x}_k  ) \\
\boldsymbol{\mu}_k (0)
\end{bmatrix},
\end{equation}
where $\boldsymbol{\pi}(\textbf{x}_k)= \begin{bmatrix} \pi_1(\textbf{x}_k) & \cdots  \pi_N (\textbf{x}_k) \end{bmatrix}^T$ for some $N>0$ and $\pi_i:\mathcal{X} \rightarrow \mathds{R}$ can be designed as some radial basis functions or neural networks.
Since predicting the future input is not of interest~\cite{korda2018linear}, we can neglect the last $m$ rows of each $\boldsymbol{\phi} ( \chi_{k+1} ) - K \boldsymbol{\phi} (\chi_{k})$ in~\eqref{eq:minimization for koopman}.
Then, let us define the first $N$ rows of $K$ as $\bar{K} = \begin{bmatrix} A & B \end{bmatrix}$, where $A \in \mathds{R}^{N \times N}$, and $B \in \mathds{R}^{N \times m}$.
%
%
%
Then, \eqref{eq:minimization for koopman} can be modified to~\cite{williams2015data} such as
\begin{equation}
\min_{A, B}
\|
\tilde{\textbf{Y}}
-A \tilde{\textbf{X}} - B \textbf{U}
\|_F^2,
\label{eq:minimization for AB}
\end{equation}
where
$\| \cdot \|_F$ denotes the Frobenius norm of a matrix,
and
\begin{eqnarray*}
\begin{array}{rcl}
\tilde{\textbf{X}} &=&
\begin{bmatrix}
\boldsymbol{\pi} ( \textbf{x}_1 ) &
\boldsymbol{\pi} ( \textbf{x}_2 ) &
\dots &
\boldsymbol{\pi} ( \textbf{x}_{M})
\end{bmatrix}  \in {\mathds R}^{N \times M},
\vspace{2mm} \\
\tilde{\textbf{Y}} &=&
\begin{bmatrix}
\boldsymbol{\pi} ( \textbf{y}_1 ) &
\boldsymbol{\pi} ( \textbf{y}_2 ) &
\dots &
\boldsymbol{\pi} ( \textbf{y}_M )
\end{bmatrix} \in {\mathds R}^{N \times M}.
%
%
\end{array}
\end{eqnarray*}
Now, the obtained Koopman operator represents the linear dynamical system, which is approximated, in the lifted space using $A$, and $B$ such that
\begin{equation}
\boldsymbol{\pi} (\textbf{x}_{k+1})  \approxeq A \boldsymbol{\pi}(\textbf{x}_{k}) + B \textbf{u}_{k}.
\label{eq:LTI koopman model}
\end{equation}
\begin{remark}
There exist other frameworks to obtain a more accurate model than the LTI model~\eqref{eq:LTI koopman model}, e.g., linear parameter varying model or input-affine model. However, this paper focuses on the uncertainty quantification of the approximated Koopman operator; therefore, the reader can refer to~\cite{iacob2022koopman} to reduce the approximation error of~\eqref{eq:LTI koopman model}.
\end{remark}
The above Koopman-based controlled linear system evolves the state in the lifted space. Thus, it is needed to reconstruct the original state $\textbf{x}_k$ from the lifted state $\boldsymbol{\pi}(\textbf{x}_{k})$. To do this, we have the optimization problem to obtain a reconstruction matrix $C$ given by
\begin{equation}
\min_{C}
\|
\textbf{X}
-C \tilde{\textbf{X}}
\|_F^2.
\label{eq:minimization for C}
\end{equation}
Then, we obtain the reconstructed original state by
\begin{equation}
\textbf{x}_k \approxeq C  \boldsymbol{\pi}(\textbf{x}_{k}).
\end{equation}
The practical method to solve the optimization problems~\eqref{eq:minimization for AB} and~\eqref{eq:minimization for C} is referred to~\cite{korda2018linear}.

However, the main challenge of using the Koopman operator is that it is not easy to design the basis (lifting) functions that provide a Koopman invariant subspace~\cite{mauroy2020koopman}.
It is known that the selection strategy of appropriate lifting function is not clearly given without domain knowledge~\cite{takeishi2017learning}.
Random sampling-based radial basis function or polynomial basis function has been used to select the lifting function~\cite{korda2018linear}, but the design process of the optimal lifting function is still open~\cite{takeishi2017learning, han2020deep, iacob2021deep}. Thus, this paper investigates the artificial intelligence method to design the lifting function for the approximation of the Koopman operator in the following section.

\section{Autoencoder-based Koopman Operator}

\subsection{Autoencoder for Lifting Function and Reconstruction}

In this paper, an AE structure is used to design the lifting function and the reconstruction function. As shown in Fig.~\ref{fig:structure}, an \emph{overcomplete} AE has a structure in which the hidden state is greater than the size of the input and output state. Since the \emph{overcomplete} AE might learn concealed useful features from the input state~\cite{bengio2009learning}, this paper adopts the AE structure.

As illustrated in Fig.~\ref{fig:structure},
the encoder $\boldsymbol{\phi}_e(\textbf{x}_k): \mathbb{R}^{n} \rightarrow \mathbb{R}^{N}$ maps the original system state $\textbf{x}_k \in \mathbb{R}^{n}$ to the lifted state $\hat{\textbf{z}}_k \in \mathbb{R}^{N}$ with higher dimension where ${N} \gg {n} $, described by the following equations~\cite{quan2022linear}:
\begin{subequations}
\begin{align}
\boldsymbol{s}^0 &= \textbf{x}_k, \\
\boldsymbol{s}^{k+1} &=
\boldsymbol{\psi}^k_e ( W_e^k \boldsymbol{s}^k + b_e^k) \text{ for }  k=0, \cdots, l_e - 1,\\
\hat{\textbf{z}}_k  &=
\boldsymbol{\psi}^{l_e}_e ( W_e^{l_e} \boldsymbol{s}^{l_e} + b_e^{l_e}),
\end{align}
\label{eq:encoder}
\end{subequations}
where $W_e^k \in \mathbb{R}^{n^e_{k+1} \times n^e_k}$ and $b_e^k \in \mathbb{R}^{n^e_{k+1}}$ are the weight matrix and bias vector for the $k$-th layer.
The function $\boldsymbol{\psi}^k_e$ is the $k$-th activation function at each layer, i.e., in the form of
$\boldsymbol{\psi}_e^k(\boldsymbol{s}^k) =
\begin{bmatrix}
\psi^k_e(s^k_1) & \cdots & \psi^k_e(s^k_{n^e_k})
\end{bmatrix} ^ T$,
where $s^k_{i}$ for $i=1,\cdots,n^e_k$ is each neuron of $\boldsymbol{s}^k$. In this case, the encoder has $l_e$ layers.
Then, the decoder $\boldsymbol{\phi}_d(\hat{\textbf{z}}_k): \mathbb{R}^{N} \rightarrow \mathbb{R}^{n}$ is represented in the similar method given by
\begin{subequations}
\begin{align}
\boldsymbol{s}^0 &= \hat{\textbf{z}}_k, \\
\boldsymbol{s}^{k+1}
&= \boldsymbol{\psi}^k_d( W_d^k \boldsymbol{s}^k + b_d^k)  \text{ for }  k=0,\cdots, l_d - 1,\\
\hat{\textbf{x}}_k
&=  W_d^{l_d} \boldsymbol{s}^{l_d} + b_d^{l_d} , \label{eq:decoder last layer}
\end{align}
\label{eq:decoder}%
\end{subequations}
where $W_d^k \in \mathbb{R}^{n^d_{k+1} \times n^d_k}$ and $b_d^k \in \mathbb{R}^{n^d_{k+1}}$
are the weight matrix and bias vector for the $k$-th layer, respectively, and $\boldsymbol{\psi}^k_d$ is the $k$-th activation function at each layer in the form of
$\boldsymbol{\psi}_d^k(\boldsymbol{s}^k) =
\begin{bmatrix}
\psi^k_d(s^k_1) & \cdots & \psi^k_d(s^k_{n^d_k})
\end{bmatrix} ^ T$,
where $s^k_{i}$ for $i=1,\cdots,n^d_k$ is each neuron of $\boldsymbol{s}^k$. In this case, the decoder has $l_d$ layers.
The loss function for training the AE includes the mean squared error and regularization, specifically:
\begin{equation}
\min_{W_{(\cdot)}^k,~b_{(\cdot)}^k}
\frac{1}{D}
\sum_{i=1}^{D}
\| \textbf{x}_i - \hat{\textbf{x}}_i \|_2^2
+ \rho \Omega_L(\boldsymbol{W}),
\end{equation}
where $\Omega_L (\boldsymbol{W})$ and $\rho$ denote the $\mathcal{L}_2$ regularization on the weights and the corresponding coefficient, respectively. Since the overcomplete AE might learn the identity matrix (e.g., just copy the input information), the regularization term is adopted to remedy the problem~\cite{bengio2009learning}.
Now, we can have the lifted state with the lifting function as $\hat{\textbf{z}}_k = \boldsymbol{\phi}_e (\textbf{x}_k )$, and the reconstruction map as $\hat{\textbf{x}}_k = \boldsymbol{\phi}_d (\hat{\textbf{z}}_k )$. The next step is to find the approximated Koopman operator with the lifted state and analyze both approximation error and reconstruction error. We will discuss it in the following subsection.
\begin{remark}
Note that this paper designs the decoder without using nonlinear activation function and bias in the last layer, i.e., $\boldsymbol{\psi}^{l_d}_d$ is not used, and $b_d^{l_d}=0$ in~\eqref{eq:decoder last layer}. Thus, the reconstructed state $\hat{\textbf{x}}_k$ can be linearly obtained from the state $\boldsymbol{s}^{l_d}$.
\label{rmk:pure linear}
\end{remark}

\subsection{Analysis of Approximation and Reconstruction Error}
\begin{figure}[t]
\includegraphics[width=\columnwidth]{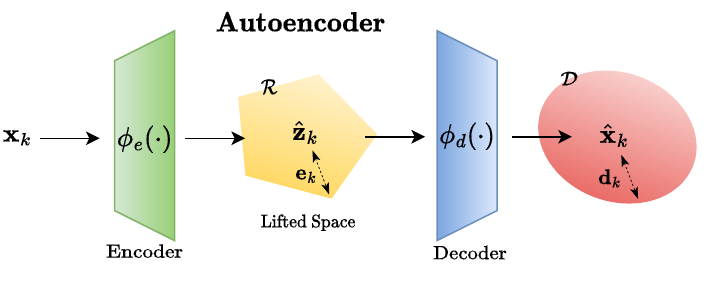}
\caption{Overall structure of the proposed method}
\label{fig:structure}
\end{figure}
Now, we have the AE to map the original state to lifted state. Thus, the Koopman operator can be obtained in the lifted space. Firstly, we collect the dataset as
\begin{equation}
\begin{split}
\textbf{X} &=
\begin{bmatrix}
\textbf{x}_1 & \textbf{x}_2 & \dots & \textbf{x}_{M}
\end{bmatrix} \in {\mathds R}^{n \times M},
\vspace{2mm} \\
\textbf{U} &=
\begin{bmatrix}
\textbf{u}_1 & \textbf{u}_2 & \dots & \textbf{u}_{M}
\end{bmatrix} \in {\mathds R}^{m \times M},
\vspace{2mm} \\
\textbf{Z}_{i}^{M} &=
\begin{bmatrix}
\boldsymbol{\phi}_e ( \textbf{x}_i ) &
\boldsymbol{\phi}_e ( \textbf{x}_2 ) &
\dots &
\boldsymbol{\phi}_e ( \textbf{x}_{M})
\end{bmatrix}  \in {\mathds R}^{N \times (M-i+1)}.
\vspace{2mm}
\end{split}
\label{eq:data collection for autoencoder}
\end{equation}
As using similar approach like~\eqref{eq:minimization for AB}, we have the optimization problem such as
\begin{equation}
\min_{\Phi, \Gamma}
\|
\textbf{Z}_{i+1}^{M}
-\Phi \textbf{Z}_{i}^{M-1} - \Gamma \textbf{U}
\|_F^2.
\label{eq:minimization for Phi Gamma}
\end{equation}
For the practical method to solve~\eqref{eq:minimization for Phi Gamma}, let us define
\begin{equation*}
\textbf{T} = \textbf{Z}_{i+1}^{M}
\begin{bmatrix} \
\textbf{Z}_{i}^{M-1} \\ \textbf{U}
\end{bmatrix}^T,~ \text{and}~~
\textbf{G} =
\begin{bmatrix}
\textbf{Z}_{i}^{M-1} \\ \textbf{U}
\end{bmatrix}
\begin{bmatrix}
\textbf{Z}_{i}^{M-1} \\ \textbf{U}
\end{bmatrix}^T.
\end{equation*}
Then, the equation can be obtained by
\begin{equation}
\textbf{T} = \mathcal{M} \textbf{G},
\label{eq:sol of normal form}
\end{equation}
where $\mathcal{M}= \begin{bmatrix} \Phi & \Gamma \end{bmatrix}$
is the approximated Koopman operator~\cite{korda2018linear}.
Then, we have
\begin{equation}
\hat{\textbf{z}}_{k+1} = \Phi \hat{\textbf{z}}_k + \Gamma \textbf{u}_k.
\label{eq:nonminal koopman model}
\end{equation}
However, the presence of modeling errors of~\eqref{eq:nonminal koopman model} is inevitable due to the approximation of the Koopman operator.
Thus, this paper considers the approximation error $\bf{w}_k$~\cite{zhang2022robust} such that
\begin{equation}
\textbf{z}_{k+1} = \Phi \textbf{z}_k + \Gamma \textbf{u}_k + \textbf{w}_k,
\label{eq:modeling fitting error}
\end{equation}
where $\textbf{z}_k \in \mathds{R}^N$ is the true lifted state following the Koopman definition~\eqref{eq:koopman definition}, and
$\textbf{w}_k \in \mathcal{W}$.
In this paper, we assume that $\mathcal{W}$ is locally bounded and includes the origin. To obtain the set $\mathcal{W}$ in the given dataset,
one can consider
\begin{equation}
\max \| \boldsymbol{\phi}_e(\textbf{x}_{k+1}) - \Phi \boldsymbol{\phi}_e (\textbf{x}_k) - \Gamma \textbf{u}_k \|_{\infty} = {\bf w}_{\text{max}},~k \in [0, M].
\end{equation}
Then, the set $\mathcal{W}$ can be obtained as $\mathcal{W}=\{ \textbf{w}_k~|~\|\textbf{w}_k\|_{\infty}  \leq \bf{w}_{\text{max}}\}$.
Let us define the state $\textbf{e}_k = \textbf{z}_k - \hat{\textbf{z}}_k$. Then, the error dynamics is given by
\begin{equation}
\textbf{e}_{k+1} = \Phi \textbf{e}_k + \textbf{w}_k.
\label{eq:error dynamics}
\end{equation}
Suppose that $\Phi$ is a strictly stable matrix~\cite{mauroy2020koopman}. Then, we can define the robust positively invariant (RPI) set as follows:
\begin{definition}[RPI set]
A set $\Omega \subset \mathds{R}^N$ is a robust positively invariant set of~\eqref{eq:error dynamics} if and only if $\Phi \Omega \oplus \mathcal{W} \subseteq \Omega$  for all $\textbf{e}_k \in \Omega$ and $\textbf{w}_k \in \mathcal{W}$, where $\oplus$ denotes the Minkowski sum\footnote{Notation: Given two sets $\mathcal{A}$ and $\mathcal{B}$, their Minkowski sum is defined by $\mathcal{A} \oplus \mathcal{B} = \{a+b | a \in \mathcal{A}, b \in \mathcal{B} \} $.}.
\end{definition}
To minimize the conservativeness, the minimal RPI (mRPI) set can be computed by
%
$\mathcal{R}_{\infty} = \bigoplus_{i=0}^{\infty} \Phi^i \mathcal{W}.
$
%
However, it is known that computing  $\mathcal{R}_{\infty}$ is impossible~\cite{rakovic2005invariant}; thus, this paper adopts the method from~\cite{rakovic2005invariant} to compute the outer approximation of the mRPI set.
\begin{theorem}[Computing RPI set]
\label{thm:RPI set}
Let the set $\mathcal{R}_s$ be
\begin{equation}
\vspace{-1mm}
\mathcal{R}_s \triangleq \bigoplus_{i=0}^{s-1} \Phi^i \mathcal{W}, ~~ \mathcal{R}_0 =\{ 0\}.
\vspace{-0.5mm}
\end{equation}
If $\mathcal{W}$ contains the origin, then there exists a finite integer $s \in \mathds{N}_+$ and a scalar $\alpha \in [0,1)$ satisfying $\Phi^s \mathcal{W} \subseteq \alpha \mathcal{W}$.
Furthermore, if $\Phi^s \mathcal{W} \subseteq \alpha \mathcal{W}$ holds, a set $\mathcal{R} (\alpha, s) \triangleq (1-\alpha)^{-1} \mathcal{R}_s$ is an outer approximation of $\mathcal{R}_{\infty}$, i.e., $\mathcal{R}_{\infty} \subseteq \mathcal{R} (\alpha, s)$.
\end{theorem}
\begin{proof}
Please refer to~\cite{rakovic2005invariant} for details.
\end{proof}
From Theorem~\ref{thm:RPI set}, we obtain the outer approximation of the mRPI set $\mathcal{R}_{\infty}$. Thus, the approximation error of~\eqref{eq:nonminal koopman model} with respect to the Koopman operator is contained in $\mathcal{R} (\alpha, s)$.

We may not need the reconstruction matrix such as $C$ in~\eqref{eq:minimization for C}, but directly use the decoder to reconstruct the true state from the lifted state.
However, the model fitting error can be propagated through the decoder layer so that there exists the reconstruction error $\textbf{x}_k - \hat{\textbf{x}}_k = \textbf{d}_k \in \mathcal{D}$ such that
\begin{equation}
\begin{split}
\textbf{x}_k - \hat{\textbf{x}}_k
&= \boldsymbol{\phi}_d ( \textbf{z}_{k})
- \boldsymbol{\phi}_d (  \hat{\textbf{z}}_{k})\\
&= \boldsymbol{\phi}_d ( \hat{\textbf{z}}_{k} + \textbf{e}_{k})
- \boldsymbol{\phi}_d (  \hat{\textbf{z}}_{k})
\end{split}
\label{eq:reconstruction error}
\end{equation}
In~\eqref{eq:reconstruction error}, it might be considered that~\eqref{eq:nonminal koopman model} can predict the true trajectories $\textbf{x}_k$ within the reconstruction error. Therefore, we need to compute the set $\mathcal{D}$ to certify the robustness of the approximated Koopman-based model.
A neural network can be analyzed by robustness certification against input uncertainties. To do this, the Lipschitz constant of the input-output map (i.e., the trained neural network) can be calculated~\cite{fazlyab2019efficient}. The details will be provided in the following subsection.

\subsection{Robustness Certification of Decoder}

In this subsection, the robustness of the decoder layer of the AE is analyzed to compute the set $\mathcal{D}$. With the set, we can guarantee that trajectory of the reconstructed state can track the trajectory of the original state within uncertainties.
To quantify the robustness of the trained AE, Lipschitz continuity is adopted~\cite{fazlyab2019efficient, quan2022linear}. Let us consider that the decoder function $\boldsymbol{\phi}_d$ is locally Lipschitz continuous if there exists $L \geq 0$ such that:
\begin{equation}
\| \boldsymbol{\phi}_d(\sigma_1) - \boldsymbol{\phi}_d(\sigma_2) \|
\leq
L\| \sigma_1-\sigma_2 \|.
\label{eq:lips}
\end{equation}
The smallest $L$ which satisfies the condition (\ref{eq:lips}) is called the Lipschitz constant $L^*$.
The Lipschitz constant $L^*$ gives an upper bound of the variations of the output of $\boldsymbol{\phi}_d$ when the input changes from $\sigma_1$ to $\sigma_2$.
Then one can obtain the following:
\begin{equation}
\| \textbf{x}_k - \hat{\textbf{x}}_k \|
= \| \boldsymbol{\phi}_d ( \hat{\textbf{z}}_{k} + \textbf{e}_{k})
- \boldsymbol{\phi}_d (  \hat{\textbf{z}}_{k}) \|
\leq
L^* \| \textbf{e}_{k} \|.
\label{eq:lips for error state}
\end{equation}
Thus, if we have $L^*$, then the reconstruction error set $\mathcal{D}$ can be computed by using $ \textbf{e}_{k} \in \mathcal{R}(\alpha, s) $.
Now, let us consider that continuous nonlinear activation functions $\boldsymbol{\psi}^k_d$ of the decoder can be interpreted as component-wise slope-restricted nonlinearity with slope at least $\alpha$ and at most $\beta$,
\begin{equation}
\alpha \leq
\frac{\boldsymbol{\psi}^k_d(\nu_1) - \boldsymbol{\psi}^k_d(\nu_2)}{\nu_1-\nu_2}
\leq \beta.
\label{eq:slope}
\end{equation}
Then by using the slope-restricted property~\eqref{eq:slope} for each neuron in the $l_d$-layer of the decoder, an incremental quadratic constraint can be used for all stacked activation functions with a diagonal weighting matrix
%
%
%
%
$T \in \mathcal{T}_n : =
\{
T = \sum_{i=1}^{n} \lambda_{ii} e_i e^T_i,
~\lambda_{ii} \geq 0
\}
\label{eq:Tn}$
%
holding
\begin{equation*}
\begin{bmatrix}
\cdot
\end{bmatrix}^T
\begin{bmatrix}
-2 \alpha \beta T & (\alpha - \beta) T \\
(\alpha - \beta) T & -2T
\end{bmatrix}
\begin{bmatrix}
\nu_1 - \nu_2 \\ \boldsymbol{\psi}^k_d(\nu_1) - \boldsymbol{\psi}^k_d(\nu_2)
\end{bmatrix}
\geq 0,
\end{equation*}
where $e_i$ is an $i$-th unit vector~\cite{fazlyab2019efficient}. Then, the smallest Lipschitz constant $L^*$ is computed by the following theorem.
\begin{theorem}[Computing $L^*$]
\label{thm:L}
Suppose the decoder has a single hidden layer, i.e., $l_d = 1$.
Consider the constrained optimization problem is given such that
\begin{equation}
L^* = \min_{L^2, T} L^2 ~~\text{s.t.}~~P_l(L^2, T) \preceq  0,~~T \in \mathcal{T}_n,
\label{eq:SDP}
\end{equation}
where
\vspace{-2mm}
\begin{equation*}
P_l(L^2, T) =
\begin{bmatrix}
-2\alpha\beta W_d^{0^T} T W_d^{0} - L^2I & (\alpha+\beta) W_d^{0^T}T \\
(\alpha+\beta) T W_d^{0} & -2 T W_d^{1^T} W_d^1
\end{bmatrix}.
\end{equation*}
If there exists $L \geq 0$ by solving~\eqref{eq:SDP}, then we can obtain the smallest Lipschitz constant $L^*$.
\end{theorem}
\begin{proof}
Please refer to~\cite{fazlyab2019efficient, quan2022linear} for details.
\end{proof}
\begin{corollary}[Computing $\mathcal{D}$]
From Theorem~\ref{thm:L}, we can obtain the smallest Lipschitz constant $L^*$ holding~\eqref{eq:lips}. Then it is immediate to find the outer approximation of the set $\mathcal{D}$ as $ L^* \mathcal{R}(\alpha, s) \supseteq \mathcal{D} $ by using~\eqref{eq:lips for error state}.
\end{corollary}
%

\section{Simulation Results}

We conduct a comparative study to validate the accuracy of the AE-based Koopman operator. We adopt the EDMD as a baseline method from~\cite{korda2018linear}, which is widely studied in the Koopman operator-based modeling. Here, for a numerical example,  we consider a forced \emph{Van der Pol} dynamics which is a stable periodic orbit so that the RPI is well defined.

\subsection{Van Der Pol and Data Collection}

For a numerical simulation, this paper considers forced \emph{Van der Pol} model as follows:

\begin{eqnarray}
\begin{array}{rcl}
\dot{x}_1 &=& 2 x_2 , \\
\dot{x}_2 &=& -0.8 x_1 - 10 x_1^2 x_2 + 2 x_2 - u.
\end{array}
\label{eq:van_der_pol}
\end{eqnarray}
%
%
In order to obtain the training dataset for the AE, we discretize the given dynamics~\eqref{eq:van_der_pol} with the fourth-order Runge-Kutta method. The sample rate is set to $T_s = 0.01s$. We have 1000 random samples for the system's initial state and simulate the 200 sample time for each initial state, i.e., 2s simulation for 1000 initial states.
the initial state is randomly obtained over the interval $\begin{bmatrix} -1 & 1 \end{bmatrix}$. The control input $u$ is also randomly selected with uniform distribution over the range $\begin{bmatrix} -1 & 1 \end{bmatrix}$.
The dataset collection setting mentioned above provides the dataset matrix $\textbf{Z}_{i}^{M-1}$, $\textbf{Z}_{i+1}^{M}$, and $\textbf{U}$.
%
%
%
The lifting functions $\boldsymbol{\pi}(\cdot)$ are used as 100 thin plate spline radial basis functions concatenated with the system state for the benchmark. The thin plate basis function is given as $\pi(\textbf{x}_k)=\| \textbf{x}_k -\textbf{x}_0 \| \cdot \text{log} (\| \textbf{x}_k -\textbf{x}_0 \|)$. The center point $\textbf{x}_0$ is randomly selected in the unit box with uniform distribution. Thus, the lifting function has a size of $N=102$.


\subsection{Autoencoder}
\begin{figure}[t]
\centering
\subfigure[Prediction of the system's first state]{
\includegraphics[width=0.8\columnwidth] {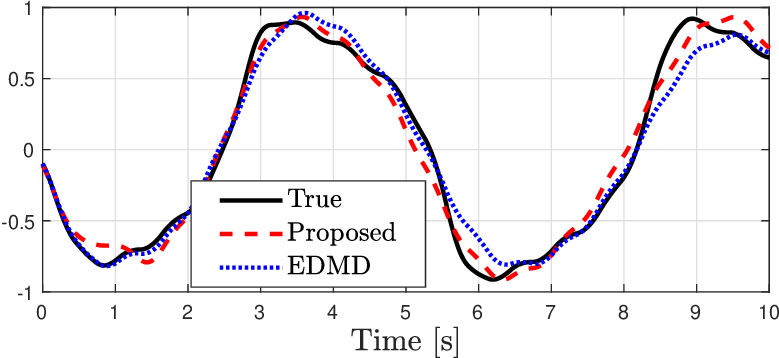}
\label{fig:fitting_x1}
}
\subfigure[Prediction of the system's second state]
{\includegraphics[width=0.8\columnwidth] {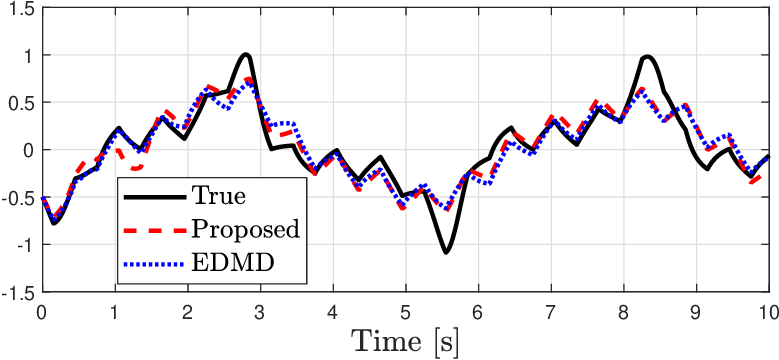}
\label{fig:fitting_x2}
}
%
\caption{Multi-step prediction performance with test set. The initial state $\textbf{x}_0$ and the input are available.}
\label{fig:fitting}
\end{figure}
%
%
%
This paper adopts a single-hidden-layer AE, i.e., the encoder has one hidden layer, and the decoder has one hidden layer ($l_e=1$, $l_d=1$). The activation function for the hidden layer of the AE is designed as the hyperbolic tangent sigmoid transfer function, i.e., $\boldsymbol{\psi}^k_e(s)= \boldsymbol{\psi}^k_d(s)= 2/(1+e^{-2s})-1$. On the other hand, the pure linear function is used for the activation function of the last layer of the decoder, as mentioned in Remark~\ref{rmk:pure linear}.
Moreover, we conduct normalization in the input layer to match the input range and extract the feature from the input data.
The training method is scaled conjugate gradient backpropagation, which performs well over a wide variety of problems~\cite{beale2010neural}. The dataset is divided three-fold, i.e., 70\% for training, 20\% for validation, and 10\% for test data. The AE is one of the unsupervised learning, so the training data for input and target of the AE is used $\textbf{X}$.
%
%

%
We trained three AEs with various sizes of neurons in the hidden layer of the AE to compare the fitting performance of trained models and the baseline method, EDMD.
We select the size of the hidden layer as $\{20, 60, 100\}$. We compare the performance in the sense of Mean Squared Error (MSE) $= (\sum_{i=0}^n \| \textbf{x}_i - \hat{\textbf{x}}_i  \|_2^2)/n$
%
%
and the maximum error (ME) for the absolute value of each system's state. Table.~\ref{table:autoencoder comparison} reports the results of the EDMD and each AE. As shown in Table.~\ref{table:autoencoder comparison},
the larger size of the hidden layer does not always bring satisfactory performance in both the MSE and the ME, and results in a large value of the Lipschitz constant. Instead, an appropriate size of the hidden layer outperforms others in our study, (i.e., $n_1^e=n_1^d=60$).

We certified the robustness of the decoder of the AE. The linear matrix inequalities (LMI) toolbox from MATLAB was used to solve~\eqref{eq:SDP}. From the trained AE, we used the weights and biases of the decoder, then obtain the smallest Lipschitz constant, $L^*$, as shown in Table.~\ref{table:autoencoder comparison}. It can be seen that the size of the decoder affects the Lipschitz constant because a large number of neurons in the hidden layer might contribute to the large uncertainties.
%

\subsection{Results}
\begin{table}[b]
\small
\caption{\small Fitting performance comparison with validation set. The method in~\cite{korda2018linear} is used for EDMD. $L^*$ is obtained by~\eqref{eq:SDP}.}
\begin{tabular}{c|cc|cc|c}
\hline
\multicolumn{1}{l|}{} & \multicolumn{2}{c|}{$x_1$}        & \multicolumn{2}{c|}{$x_2$}        & \multirow{2}{*}{$L^*$} \\ \cline{2-5}
& ME       & MSE             & ME       & MSE             &                                         \\ \hline
EDMD        & 0.350          & 0.013          & 0.497          & 0.028          & -                                       \\
$n^{e}_1=n^{d}_1=20$            & 0.753          & 0.108          & 0.698          & 0.081          & \textbf{1.6889}                                  \\
$n^{e}_1=n^{d}_1=60$            & \textbf{0.203} & \textbf{0.009} & \textbf{0.460} & 0.026          & 2.1974                        \\
$n^{e}_1=n^{d}_1=100$           & 0.255          & 0.010          & 0.465          & \textbf{0.025} & 2.7883                                  \\
\hline
\end{tabular}
\label{table:autoencoder comparison}
\end{table}
%
%
%

\begin{figure}[t]
\includegraphics[width=1\columnwidth]{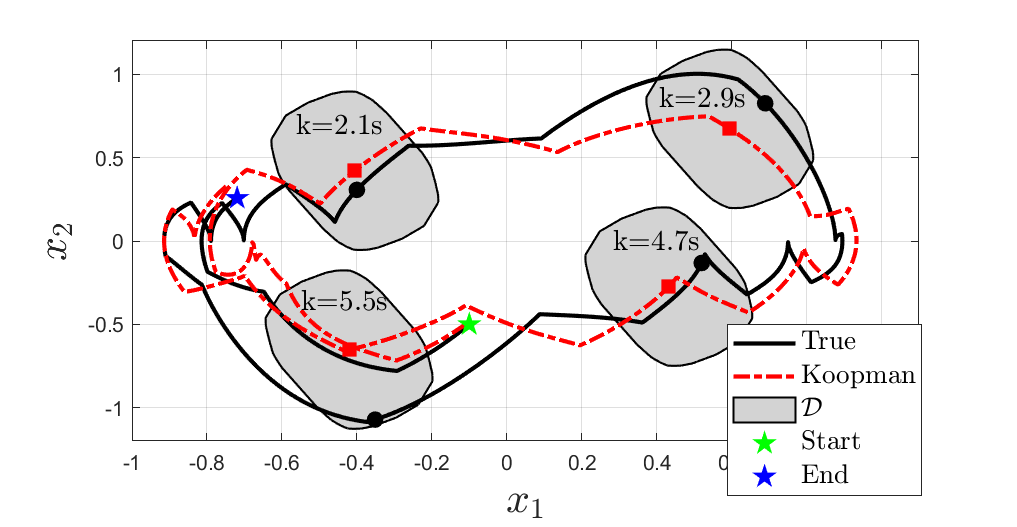}
\caption{\small State trajectory for $k \in [0, 7]$ in phase plane: The initial state (i.e., green pentagram) is $\textbf{x}_0^T=[-0.1~-0.5]^T$. Gray area shows set $\mathcal{D}$ at each k=2.1s, 2.9s, 4.7s, and 5.5s.
%
%
}
\label{fig:set D}
\end{figure}

In order to evaluate the performance of the proposed method, we exert a specific shape of input into the model~\eqref{eq:van_der_pol}. The system input $u$ is designed as square shaped wave signal whose magnitude is one and frequency is 0.3s~\cite{korda2018linear}. We set the initial state as $\textbf{x}_0^T = \begin{bmatrix}-0.1 &-0.5\end{bmatrix}^T$ to validate the accuracy of the AE-based Koopman model.

The result of each system state prediction is obtained with only the initial state $\textbf{x}_0$ and input, as shown in Fig.~\ref{fig:fitting}.
In Fig.~\ref{fig:fitting}, the solid black line is the true trajectory of the system generated by the input $u$. The red dashed line stands for the result from the AE. The AE is adopted with $n_1^e=n_1^d=60$ because overall fitness is better than other AEs'.
The blue dotted line represents the EDMD result. It can be seen that the AE has better fitting performance compared to the EDMD.
%
%
Moreover, we compared quantitatively in Table.~\ref{table:autoencoder comparison} that the AE with $n_1^e=n_1^d=60$ had less ME and MSE than the EDMD has.
%
%

We calculated the set $\mathcal{D}$ and evaluated its validity. As is illustrated in Fig.~\ref{fig:set D}, the solid black line stands for the trajectory of the true state, the dotted red line is generated by the proposed method, and the gray set represents the set $\mathcal{D}$. We sample four points of each trajectory as circle markers for the true state and square markers for the proposed method. With sample points, it can be shown that the true state stays in the set $\mathcal{D}$ centered by the reconstructed state. Especially, the true state is included in the set $\mathcal{D}$ even though the prediction error is large at time $k=5.5s$.
%

\section{Conclusion and Future Work}

In this paper, we proposed the method for uncertainty quantification of the autoencoder-based Koopman operator. This paper used the AE to design the lifting basis functions. It also considered the approximation error resulting from the finite-dimensional Koopman approximation. To this end, we computed the RPI set for the approximated Koopman model against approximation uncertainty. Moreover, we considered the reconstruction error propagated by the approximation error through the decoder. In order to acquire the reconstruction uncertainty set, the decoder was analyzed by robustness certification against the approximation error using the Lipschitz constant. In the simulation study, we used the forced \emph{Van der Pol} model to validate the effectiveness of the proposed method. The simulation results confirmed that the true state stayed in the uncertainty set centered by the reconstructed state. In practice, there exist state uncertainties due to process or measurement noise. Thus, the propagation of state uncertainties in the encoder will be considered in future work.


%
\bibliographystyle{IEEEtran}        
\bibliography{IEEEabrv,BIB_koopman_uncertainty}   

\end{document}